\documentclass[conference]{IEEEtran}
\ifCLASSINFOpdf
\else
\fi
\usepackage{graphicx}
\usepackage{fullpage,graphicx}
\usepackage{stmaryrd}
\usepackage{amsmath}
\usepackage{amssymb}
\usepackage{dsfont}
\usepackage{color}
\usepackage{subfigure}
\usepackage[ruled,vlined]{algorithm2e}

\usepackage[standard]{ntheorem}

\begin{document}
%
\title{Summary of Topological Study of Chaotic CBC Mode of Operation}

\author{\IEEEauthorblockN{Abdessalem Abidi}
\IEEEauthorblockA{Electronics and Microelectronics Lab.\\
Faculty of Sciences of Monastir\\
University of Monastir, Tunisia\\
Email: abdessalemabidi9@gmail.com}

\and
\IEEEauthorblockN{Samar Tawbi}
\IEEEauthorblockA{Faculty of Science\\
Lebanese University, Beirut\\
Lebanon}
\and
\IEEEauthorblockN{Christophe Guyeux}
\IEEEauthorblockA{FEMTO-ST Institute, UMR 6174 CNRS,\\ University of Franche-Comt\'e, France}
\and
\IEEEauthorblockN{Belgacem Bouall\`egue\\ and Mohsen Machhout}
\IEEEauthorblockA{Electronics and Microelectronics Lab.\\
Faculty of Sciences of Monastir\\
University of Monastir, Tunisia}}


%


\maketitle

\begin{abstract}
In cryptography, block ciphers are the most fundamental elements in many symmetric-key encryption systems. The Cipher Block Chaining, denoted CBC, presents one of the most famous mode of operation that uses a block cipher to provide confidentiality or authenticity. In this research work, we intend to summarize our results that have been detailed in our previous series of articles. The goal of this series has been to obtain a complete topological study of the CBC block cipher mode of operation after proving his chaotic behavior according to the reputed definition of Devaney.
\end{abstract}


%
\IEEEpeerreviewmaketitle

\section{Introduction}

Cryptography is the mathematics of secret codes. Block ciphers are the most fundamental elements in many symmetric-key encryption systems, it means that both encryption and decryption utilize the same key. Block cipher algorithm divides the message into separate blocks of fixed size, and encrypts/decrypts each block individually. The encrypted block and the original one have the same size. 
Block ciphers provide confidentiality that is used in a large variety of applications, such as: protection of the secrecy of login passwords, email messages, video transmissions and many other applications.

At the programming level, it is not sufficient to put anyhow a block cipher algorithm. In fact, this latter can be used in various ways depending on their specific needs. These ways are called the block cipher modes of operation. There are several modes and each one of them possesses own characteristics in addition to its specific security properties. In this paper, the mode on which we will focus is the cipher block chaining one, and we will study it according to chaos.

The chaos theory that we consider in our various research works is the Devaney's topological one~\cite{Devaney}. In addition to being recognized as one of the best mathematical definition of chaos, this theory offers a framework with qualitative and quantitative tools to evaluate the notion of unpredictability~\cite{bahi2011efficient}.
As an application of our fundamental results, we are interested in the area of information safety and security, and more precisely we focus on symmetric-key encryption systems.
In this research work, which is a summary of our previous papers, we intend to preview the different results, which have been detailed respectively in ~\cite{Abdessalem2016}, ~\cite{abidi:hal-01264170}, and ~\cite{abidi:hal-01312476}. These results allowed us to deepen a complete topological study of the CBC mode of operation after proving its chaotic behavior according to Devaney.

The remainder of this research work is organized as follows. In Section~\ref{section:BASIC RECALLS}, we will recall some basic definitions concerning chaos and  CBC mode of operation.
Section~\ref{sec:proof} is devoted to sum up the various demonstrations that allowed us to proof the chaotic behaviour of the CBC mode. Moreover, new cases will be studied here. In Section~\ref{section:Quantitative measures}, quantitative topological properties for chaotic CBC mode of operation will be summarized, while 
Section~\ref{top mixing} resumes the main contribution that has been detailed in our last article ~\cite{abidi:hal-01312476}. This research work ends by a conclusion section in which our contribution is recalled and some intended future work are proposed.

\section{Basic Recalls}
\label{section:BASIC RECALLS}
\subsection{Devaney's Chaotic Dynamical Systems}
\label{subsec:Devaney}
In the remainder of this paper, $m_{n}$ denotes the $n^{th}$ block message of a sequence $S$ while $m^{j}$ stands for the $j-th$ bit of integer of the block message $m\in \llbracket 0, 2^\mathsf{N}-1 \rrbracket$, expressed in the binary numeral system and $x_{i}$ stands for the $i^{th}$ component of a vector $x$. 

$\mathcal{X}^\mathds{N}$ is the set of all sequences whose elements belong to $\mathcal{X}$.

$f^{\circ k}=f\circ ...\circ f$ is for the $k^{th}$ composition of a function $f$.
$\mathds{N}$ is the set of natural (non-negative) numbers, while $\mathds{N}^*$ stands for the positive integers $1, 2, 3, \hdots$ 

Finally, the following
notation is used: $\llbracket1;N\rrbracket=\{1,2,\hdots,N\}$.

Consider a topological space $(\mathcal{X},\tau)$, where $\tau$ represents a family of subsets of $\mathcal{X}$, and a continuous function $f :
\mathcal{X} \rightarrow \mathcal{X}$ on $(\mathcal{X},\tau)$ ~\cite{abidi:hal-01312476}.

\begin{definition}
The function $f$ is \emph{topologically transitive} if, for any pair of nonempty open sets
$\mathcal{U},\mathcal{V} \subset \mathcal{X}$, there exists an integer $k>0$ such that $f^{\circ k}(\mathcal{U}) \cap \mathcal{V} \neq
\varnothing$.
\end{definition}

\begin{definition}
An element $x$ is a \emph{periodic point} for $f$ of period $n\in \mathds{N}$, $n>1$,
if $f^{\circ n}(x)=x$ and $f^{\circ k}(x) \neq x, 1\le k\le n$. 
\end{definition}

\begin{definition}
$f$ is  \emph{regular} on $(\mathcal{X}, \tau)$ if the set of periodic
points for $f$ is dense in $\mathcal{X}$: for any point $x$ in $\mathcal{X}$,
any neighborhood of $x$ contains at least one periodic point.
\end{definition}

\begin{definition}
\label{sensitivity} The function $f$ has \emph{sensitive dependence on initial conditions} on the metric space $(\mathcal{X},d)$
if there exists $\delta >0$ such that, for any $x\in \mathcal{X}$ and any
neighborhood $\mathcal{V}$ of $x$, there exist $y\in \mathcal{V}$ and $n > 0$ such that
the distance $d$ between the results of their $n^{th}$ composition, $f^{\circ n}(x)$ and $f^{\circ n}(y)$, is greater than $\delta$:
$$d\left(f^{\circ n}(x), f^{\circ n}(y)\right) >\delta .$$
$\delta$ is called the \emph{constant of sensitivity} of $f$.
\end{definition}

\begin{definition}[Devaney's formulation of chaos~\cite{Devaney}]\label{def:dev}
The function $f$ is  \emph{chaotic} on a metric space $(\mathcal{X},d)$ if $f$ is regular,
topologically transitive, and has sensitive dependence on initial conditions.
\end{definition}

Banks \emph{et al.} have proven in~\cite{Banks92} that when $f$ is regular and transitive on a metric space $(\mathcal{X}, d)$, then $f$ has the property of sensitive dependence on initial conditions. This is why chaos can be formulated too in a topological space $(\mathcal{X}, \tau)$: in that situation, chaos is obtained when $f$ is regular and
topologically transitive.
Note that the transitivity property is often obtained as a consequence of the strong transitivity one, which is defined below.

\begin{definition}
\label{def:strongTrans}
$f$ is \emph{strongly transitive} on $(\mathcal{X},d)$ if, for all point $x,y \in \mathcal{X}$ and for all neighborhood $\mathcal{V}$ of $x$, it exists $n \in \mathds{N}$ and $x'\in \mathcal{V}$ such that $f^n(x')=y$. 
\end{definition}

Finally, a function \emph{f} has a constant of \emph{expansivity} equal to $\varepsilon $ if an arbitrarily small error on any initial condition is \emph{always} magnified until $\varepsilon $~\cite{gb11:bc}. Mathematically speaking,

\begin{definition}\label{def:expan}
  The function $f$  is said to have the property of \emph{expansivity}
 if $\exists \varepsilon >0,$ $\forall x \neq y,$ $\exists n \in \mathbb{N},$ $d(f^{n}(x),f^{n}(y)) \geqslant \varepsilon$.
 
Then, $\varepsilon $ is the \emph{constant of expansivity} of \emph{f}. We also say that \emph{f} is $\varepsilon$-expansive.
\end{definition}

\begin{definition} 
\label{def:topMixing}
A discrete dynamical system is said \emph{topologically mixing} if and only if, for any couple of disjoint open set $\mathcal{U},\mathcal{V} \neq \varnothing$, there exists an integer $n_0\in \mathds{N}$ such that, for all $n > n_0$, $f^{\circ n}(\mathcal{U}) \cap \mathcal{V} \neq
\varnothing$.
\end{definition}

\subsection{CBC mode characteristics}
In cryptography, Cipher Block Chaining is a block cipher mode that provides confidentiality but not message integrity. Similar to some other modes, the input to the encryption processes of the CBC mode includes not only a plaintext, but also a data block called the initialization vector, which is denoted IV.

 In particular, the CBC mode offers a solution to most of the problems presented by the Electronic Code Book (ECB). Indeed, thanks to this mode, the encryption will depend not only on the plaintext, but also on all preceding blocks. More precisely, each block of plaintext is XORed immediately with the previous cipher text block before being encrypted (\textit{i.e.}, the binary operator XOR is applied between two stated blocks). 

 For the first block, the initialization vector acts as the previous cipher text block. For decryption, we proceed in the same way, but this time, we start from the encrypted text to obtain the original one using now the decryption algorithm instead of the encryption function, see Figure 1.
 \begin{figure}[!ht]
    \centering
 \subfigure[CBC encryption mode]{\label{fig:CBCenc}
        \includegraphics[scale=0.5]{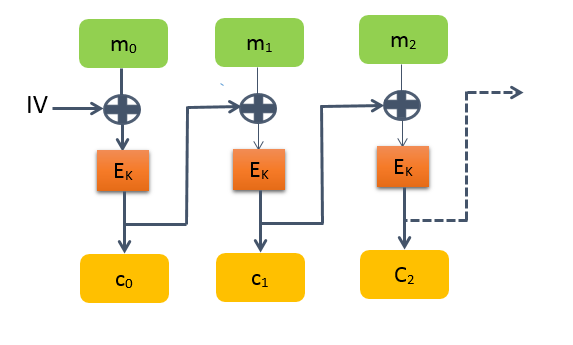}}     \subfigure[CBC decryption mode]{\includegraphics[scale=0.5]{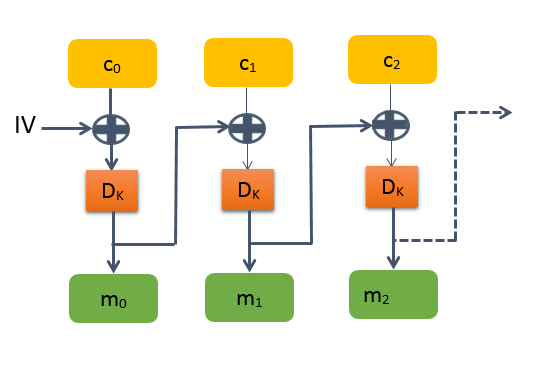}}
    \caption{CBC mode of operation}
     \label{fig:CBC}
\end{figure}

 The IV needs not to be secret; however, for any particular execution of the encryption process, it must be unpredictable. In the next section, we will summarize the results of our first article ~\cite{Abdessalem2016}.
\section{Proving Chaotic behavior of CBC mode}
\label{sec:proof}
In this section, we prove that CBC mode of operation behaves as Devaney's topological chaos if the iteration function used is the vectorial Boolean negation. This function has been chosen here, but the process remains general and other iterate functions g can be used. The sole condition is to prove that $G_g$ satisfies the Devaney's chaos property. To do this, we have began by modeling CBC mode as a dynamical system.
 
\subsection{Modeling CBC as dynamical system}
Our modeling follows a same canvas than what has been done for hash functions~\cite{bg10:ij,gb11:bc} or pseudorandom number generation~\cite{bfgw11:ij}.
Let us consider the CBC mode of operation with a keyed encryption function $\varepsilon_k:\mathds{B}^\mathsf{N} \rightarrow \mathds{B}^\mathsf{N} $ depending on a secret key $k$, where $\mathsf{N}$ is the size for the block cipher, and $\mathcal{D}_k:\mathds{B}^\mathsf{N} \rightarrow \mathds{B}^\mathsf{N} $ is the associated  decryption function, which is such that $\forall k, \varepsilon_k \circ \mathcal{D}_k$ is the identity function. We define  
the Cartesian product $\mathcal{X}=\mathds{B}^\mathsf{N}\times\mathcal{S}_\mathsf{N}$, where:
\begin{itemize}
\item $\mathds{B} = \{0,1\}$ is the set of Boolean values,
\item $\mathcal{S}_\mathsf{N} = \llbracket 0, 2^\mathsf{N}-1\rrbracket^\mathds{N}$, the set of infinite sequences of natural integers bounded by $2^\mathsf{N}-1$, or the set of infinite $\mathsf{N}$-bits block messages,
\end{itemize}
in such a way that $\mathcal{X}$ is constituted by couples of internal states of the mode of operation together with sequences of block messages.
Let us consider the initial function:
$$\begin{array}{cccc}
 i:& \mathcal{S}_\mathsf{N} & \longrightarrow & \llbracket 0, 2^\mathsf{N}-1 \rrbracket \\
 & (m^i)_{i \in \mathds{N}} & \longmapsto & m^0
\end{array}$$
that returns the first block of a (infinite) message, and the shift function:
$$\begin{array}{cccc}
 \sigma:& \mathcal{S}_\mathsf{N} & \longrightarrow & \mathcal{S}_\mathsf{N} \\
 & (m^0, m^1, m^2, ...) & \longmapsto & (m^1, m^2, m^3, ...)
\end{array}$$
which removes the first block of a message. Let $m_j$ be the $j$-th bit of integer, or block message, $m\in \llbracket 0, 2^\mathsf{N}-1 \rrbracket$, expressed in the binary numeral system, and when counting from the left. We define:
$$\begin{array}{cccc}
F_f:& \mathds{B}^\mathsf{N}\times \llbracket 0, 2^\mathsf{N}-1 \rrbracket & \longrightarrow & \mathds{B}^\mathsf{N}\\
 & (x,m) & \longmapsto & \left(x_j \overline{m_j} + f(x)_j {m_j}\right)_{j=1..\mathsf{N}} 
\end{array}$$
This function returns the inputted binary vector $x$, whose $m_j$-th components $x_{m_j}$ have been replaced by $f(x)_{m_j}$, for all $j=1..\mathsf{N}$ such that $m_j=1$. In case where $f$ is the vectorial negation, this function will correspond to one XOR between the clair text and the previous encrypted state.  
So the CBC mode of operation can be rewritten as the following dynamical system:
\begin{equation}
\label{eq:sysdyn}
\left\{
\begin{array}{ll}
X^0 = & (IV,m)\\
X^{n+1} = & \left(\mathcal{E}_k \circ F_{f_0} \left( i(X_1^n), X_2^n\right), \sigma (X_1^n)\right)
\end{array}
\right.
\end{equation}
For any given $g:\llbracket 0, 2^\mathsf{N}-1\rrbracket \times \mathds{B}^\mathsf{N} \longrightarrow \mathds{B}^\mathsf{N}$, we denote $G_g(X) = \left(g(i(X_1),X_2);\sigma (X_1)\right)$ (when $g = \mathcal{E}_k\circ F_{f_0}$, we obtain one cypher block of the CBC, as depicted in Figure~\ref{fig:CBC}). So the recurrent relation of Eq.\eqref{eq:sysdyn} can be rewritten in a condensed way, as follows.
\begin{equation}
X^{n+1} = G_{\mathcal{E}_k\circ F_{f_0}} \left(X^n\right) .
\end{equation}
With such a rewriting, one iterate of the discrete dynamical system above corresponds exactly to one cypher block in the CBC mode of operation. Note that the second component of this system is a subshift of finite type, which is related to the symbolic dynamical systems known for their relation with chaos~\cite{lind1995introduction}.
We now define a distance on $\mathcal{X}$ as follows: $d((x,m);(\check{x},\check{m})) = d_e(x,\check{x})+d_m(m,\check{m})$, where:
$$\left\{\begin{array}{ll}
d_e(x,\check{x})  & = \sum_{k=1}^\mathsf{N} \delta (x_k,\check{x}_k)  \\
&\\
d_m(m,\check{m}) & = \displaystyle{\dfrac{9}{\mathsf{N}} \sum_{k=1}^\infty \dfrac{\sum_{i=1}^\mathsf{N} \left|m_i - \check{m}_i\right|}{10^k}} .
\end{array}\right.$$
This distance has been introduced to satisfy the following requirements:
\begin{itemize}
\item The integral part between two points $X,Y$ of the phase space $\mathcal{X}$ corresponds to the number of binary components that are different between the two internal states $X_1$ and $Y_1$.
\item The $k$-th digit in the decimal part of the distance between $X$ and $Y$ is equal to 0 if and only if the $k$-th blocks of messages $X_2$ and $Y_2$ are equal. This desire is at the origin of the normalization factor $\dfrac{9}{\mathsf{N}}$.
\end{itemize}

\subsection{Proof of chaos}
As mentioned in Definition~\ref{def:dev},  a function $f$ is  \emph{chaotic} on $(\mathcal{X},\tau)$ if $f$ is regular and
topologically transitive. 
We have began in~\cite{Abdessalem2016} by stating some propositions that are primarily required in order to proof the chaotic behavior of the CBC mode of operation.

\begin{proposition}
\label{prop:transitivity}
Let $g=\mathcal{E}_{\kappa} \circ F_{f_0}$, where $\mathcal{E}_{\kappa}$ is a given keyed block cipher and $f_0:\mathds{B}^\mathsf{N} \longrightarrow \mathds{B}^\mathsf{N}$, $(x_1,...,x_\mathsf{N}) \longmapsto (\overline{x_1},...,\overline{x_\mathsf{N}})$ is the Boolean vectorial negation.
We consider the directed graph $\mathcal{G}_g$, where:
\begin{itemize}
\item vertices are all the $\mathsf{N}$-bit words.
\item there is an edge $m \in \llbracket 0, 2^{\mathsf{N}}-1 \rrbracket$ from $x$ to $\check{x}$ if and only if $g(m,x)=\check{x}$.
\end{itemize}
If $\mathcal{G}_g$ is strongly connected, then $G_g$ is strongly transitive.
\end{proposition}

We have then proven that,
\begin{proposition}
\label{prop:regularity}
If $\mathcal{G}_g$ is strongly connected, then $G_g$ is regular.
\end{proposition}

According to Propositions~\ref{prop:transitivity} and~\ref{prop:regularity}, we can conclude that, depending on $g$, if the directed graph $\mathcal{G}_g$ is strongly connected, then the CBC mode of operation is chaotic according to Devaney, as established in our previous research work~\cite{Abdessalem2016}.

In this article and for illustration purpose, we have also given some examples of keyed block ciphers, which can be used by the CBC mode, and which can lead to a chaotic behavior for this mode when they have a strongly connected directed graph. These examples are taken from so-called transposition cipher methods. Among these examples, we considered the Caesar shift one.

\subsection{Caesar shift case}

The Caesar shift case is considered as one of the simplest and most widely known substitution cipher. 
In this method, each symbol in the plaintext is replaced by a symbol some fixed number of positions down the given alphabet. Translated in the $\mathsf{N}$ binary digits set of integers, this cypher can be written as follows:
$$
\begin{array}{cccc}
\mathcal{E}_k(x):& \llbracket 0, 2^\mathsf{N}-1 \rrbracket& \longrightarrow &\llbracket 0, 2^\mathsf{N}-1 \rrbracket\\
& x & \longmapsto & x+k \mod 2^\mathsf{N} \end{array}
$$

$$
\begin{array}{cccc}
\mathcal{D}_k(x):& \llbracket 0, 2^\mathsf{N}-1 \rrbracket& \longrightarrow &\llbracket 0, 2^\mathsf{N}-1 \rrbracket\\
& x & \longmapsto & x-k \mod 2^\mathsf{N} \end{array}
$$
\noindent where $k$ is the shift value acting as secret key. We will now show through examples that the CBC mode of operation embedding the Caesar shift can behave either chaotically or not, depending on $k$ and $\mathsf{N}$.
The two following tables~\ref{tab:Caesar} and~\ref{tab:Caesar1}  contain the $g(x,m)$ values for a shift of 1 and 2 respectively, in Caesar cipher over $3$-bit blocks. 
\begin{table}[]
\centering
\begin{tabular}{cc}
\begin{tabular}{c|c|c|c}
$x$ & $m$ & $F_{f_0}(x,m)$ & $g(m,x) = \mathcal{E}_k \circ F_{f_0}(x,m)$ \\
\hline
0 (0,0,0) & 0 (0,0,0) & 0 (0,0,0) & 1 \\
0 (0,0,0) & 1 (0,0,1) & 1 (0,0,1) & 2 \\
0 (0,0,0) & 2 (0,1,0) & 2 (0,1,0) & 3 \\
0 (0,0,0) & 3 (0,1,1) & 3 (0,1,1) & 4 \\
0 (0,0,0) & 4 (1,0,0) & 4 (1,0,0) & 5 \\
0 (0,0,0) & 5 (1,0,1) & 5 (1,0,1) & 6 \\
0 (0,0,0) & 6 (1,1,0) & 6 (1,1,0) & 7 \\
0 (0,0,0) & 7 (1,1,1) & 7 (1,1,1) & 0 \\
\hline
1 (0,0,1) & 0 (0,0,0) & 1 (0,0,1) & 2 \\
1 (0,0,1) & 1 (0,0,1) & 0 (0,0,0) & 1 \\
1 (0,0,1) & 2 (0,1,0) & 3 (0,1,1) & 4 \\
1 (0,0,1) & 3 (0,1,1) & 2 (0,1,0) & 3 \\
1 (0,0,1) & 4 (1,0,0) & 5 (1,0,1) & 6 \\
1 (0,0,1) & 5 (1,0,1) & 4 (1,0,0) & 5 \\
1 (0,0,1) & 6 (1,1,0) & 7 (1,1,1) & 0 \\
1 (0,0,1) & 7 (1,1,1) & 6 (1,1,0) & 7 \\
\hline
2 (0,1,0) & 0 (0,0,0) & 2 (0,1,0) & 3 \\
2 (0,1,0) & 1 (0,0,1) & 3 (0,1,1) & 4 \\
2 (0,1,0) & 2 (0,1,0) & 0 (0,0,0) & 1 \\
2 (0,1,0) & 3 (0,1,1) & 1 (0,0,1) & 2 \\
2 (0,1,0) & 4 (1,0,0) & 6 (1,1,0) & 7 \\
2 (0,1,0) & 5 (1,0,1) & 7 (1,1,1) & 0 \\
2 (0,1,0) & 6 (1,1,0) & 4 (1,0,0) & 5 \\
2 (0,1,0) & 7 (1,1,1) & 5 (1,0,1) & 6 \\
\hline
3 (0,1,1) & 0 (0,0,0) & 3 (0,1,1) & 4 \\
3 (0,1,1) & 1 (0,0,1) & 2 (0,1,0) & 3 \\
3 (0,1,1) & 2 (0,1,0) & 1 (0,0,1) & 2 \\
3 (0,1,1) & 3 (0,1,1) & 0 (0,0,0) & 1 \\
3 (0,1,1) & 4 (1,0,0) & 7 (1,0,0) & 0 \\
3 (0,1,1) & 5 (1,0,1) & 6 (1,1,0) & 7 \\
3 (0,1,1) & 6 (1,1,0) & 5 (1,1,0) & 6 \\
3 (0,1,1) & 7 (1,1,1) & 4 (1,0,0) & 5 \\
\hline
4 (1,0,0) & 0 (0,0,0) & 4 (1,0,0) & 5 \\
4 (1,0,0) & 1 (0,0,1) & 5 (1,1,0) & 6 \\
4 (1,0,0) & 2 (0,1,0) & 6 (1,1,0) & 7 \\
4 (1,0,0) & 3 (0,1,1) & 7 (1,0,0) & 0 \\
4 (1,0,0) & 4 (1,0,0) & 0 (0,0,0) & 1 \\
4 (1,0,0) & 5 (1,0,1) & 1 (0,0,1) & 2 \\
4 (1,0,0) & 6 (1,1,0) & 2 (0,1,0) & 3 \\
4 (1,0,0) & 7 (1,1,1) & 3 (0,1,1) & 4 \\
\hline
5 (1,0,1) & 0 (0,0,0) & 5 (1,1,0) & 6 \\
5 (1,0,1) & 1 (0,0,1) & 4 (1,0,0) & 5 \\
5 (1,0,1) & 2 (0,1,0) & 7 (1,0,0) & 0 \\
5 (1,0,1) & 3 (0,1,1) & 6 (1,1,0) & 7 \\
5 (1,0,1) & 4 (1,0,0) & 1 (0,0,1) & 2 \\
5 (1,0,1) & 5 (1,0,1) & 0 (0,0,0) & 1 \\
5 (1,0,1) & 6 (1,1,0) & 3 (0,1,1) & 4 \\
5 (1,0,1) & 7 (1,1,1) & 2 (0,1,0) & 3 \\
\hline
6 (1,1,0) & 0 (0,0,0) & 6 (1,1,0) & 7 \\
6 (1,1,0) & 1 (0,0,1) & 7 (1,0,0) & 0 \\
6 (1,1,0) & 2 (0,1,0) & 4 (1,0,0) & 5 \\
6 (1,1,0) & 3 (0,1,1) & 5 (1,1,0) & 6 \\
6 (1,1,0) & 4 (1,0,0) & 2 (0,1,0) & 3 \\
6 (1,1,0) & 5 (1,0,1) & 3 (0,1,1) & 4 \\
6 (1,1,0) & 6 (1,1,0) & 0 (0,0,0) & 1 \\
6 (1,1,0) & 7 (1,1,1) & 1 (0,0,1) & 2 \\
\hline
7 (1,1,1) & 0 (0,0,0) & 7 (1,0,0) & 0 \\
7 (1,1,1) & 1 (0,0,1) & 6 (1,1,0) & 7 \\
7 (1,1,1) & 2 (0,1,0) & 5 (1,1,0) & 6 \\
7 (1,1,1) & 3 (0,1,1) & 4 (1,0,0) & 5 \\
7 (1,1,1) & 4 (1,0,0) & 3 (0,1,1) & 4 \\
7 (1,1,1) & 5 (1,0,1) & 2 (0,1,0) & 3 \\
7 (1,1,1) & 6 (1,1,0) & 1 (0,0,1) & 2 \\
7 (1,1,1) & 7 (1,1,1) & 0 (0,0,0) & 1 \\
 \end{tabular}

\\

\end{tabular}
\caption{$g(x,m)$ for $\mathsf{N}=3$, $k=1$ }
\label{tab:Caesar}
\end{table}
\begin{table}[]
\centering
\begin{tabular}{cc}
\begin{tabular}{c|c|c|c}
$x$ & $m$ & $F_{f_0}(x,m)$ & $g(m,x) = \mathcal{E}_k \circ F_{f_0}(x,m)$ \\
\hline
0 (0,0,0) & 0 (0,0,0) & 0 (0,0,0) & 2 \\
0 (0,0,0) & 1 (0,0,1) & 1 (0,0,1) & 3 \\
0 (0,0,0) & 2 (0,1,0) & 2 (0,1,0) & 4 \\
0 (0,0,0) & 3 (0,1,1) & 3 (0,1,1) & 5 \\
0 (0,0,0) & 4 (1,0,0) & 4 (1,0,0) & 6 \\
0 (0,0,0) & 5 (1,0,1) & 5 (1,0,1) & 7 \\
0 (0,0,0) & 6 (1,1,0) & 6 (1,1,0) & 0 \\
0 (0,0,0) & 7 (1,1,1) & 7 (1,1,1) & 1 \\
\hline
1 (0,0,1) & 0 (0,0,0) & 1 (0,0,1) & 3 \\
1 (0,0,1) & 1 (0,0,1) & 0 (0,0,0) & 2 \\
1 (0,0,1) & 2 (0,1,0) & 3 (0,1,1) & 5 \\
1 (0,0,1) & 3 (0,1,1) & 2 (0,1,0) & 4 \\
1 (0,0,1) & 4 (1,0,0) & 5 (1,0,1) & 7 \\
1 (0,0,1) & 5 (1,0,1) & 4 (1,0,0) & 6 \\
1 (0,0,1) & 6 (1,1,0) & 7 (1,1,1) & 1 \\
1 (0,0,1) & 7 (1,1,1) & 6 (1,1,0) & 0 \\
\hline
2 (0,1,0) & 0 (0,0,0) & 2 (0,1,0) & 4 \\
2 (0,1,0) & 1 (0,0,1) & 3 (0,1,1) & 5 \\
2 (0,1,0) & 2 (0,1,0) & 0 (0,0,0) & 2 \\
2 (0,1,0) & 3 (0,1,1) & 1 (0,0,1) & 3 \\
2 (0,1,0) & 4 (1,0,0) & 6 (1,1,0) & 0 \\
2 (0,1,0) & 5 (1,0,1) & 7 (1,1,1) & 1 \\
2 (0,1,0) & 6 (1,1,0) & 4 (1,0,0) & 6 \\
2 (0,1,0) & 7 (1,1,1) & 5 (1,0,1) & 7 \\
\hline
3 (0,1,1) & 0 (0,0,0) & 3 (0,1,1) & 5 \\
3 (0,1,1) & 1 (0,0,1) & 2 (0,1,0) & 4 \\
3 (0,1,1) & 2 (0,1,0) & 1 (0,0,1) & 3 \\
3 (0,1,1) & 3 (0,1,1) & 0 (0,0,0) & 2 \\
3 (0,1,1) & 4 (1,0,0) & 7 (1,0,0) & 1 \\
3 (0,1,1) & 5 (1,0,1) & 6 (1,1,0) & 0 \\
3 (0,1,1) & 6 (1,1,0) & 5 (1,1,0) & 7 \\
3 (0,1,1) & 7 (1,1,1) & 4 (1,0,0) & 6 \\
\hline
4 (1,0,0) & 0 (0,0,0) & 4 (1,0,0) & 6 \\
4 (1,0,0) & 1 (0,0,1) & 5 (1,1,0) & 7 \\
4 (1,0,0) & 2 (0,1,0) & 6 (1,1,0) & 0 \\
4 (1,0,0) & 3 (0,1,1) & 7 (1,0,0) & 1 \\
4 (1,0,0) & 4 (1,0,0) & 0 (0,0,0) & 2 \\
4 (1,0,0) & 5 (1,0,1) & 1 (0,0,1) & 3 \\
4 (1,0,0) & 6 (1,1,0) & 2 (0,1,0) & 4 \\
4 (1,0,0) & 7 (1,1,1) & 3 (0,1,1) & 5 \\
\hline
5 (1,0,1) & 0 (0,0,0) & 5 (1,1,0) & 7 \\
5 (1,0,1) & 1 (0,0,1) & 4 (1,0,0) & 6 \\
5 (1,0,1) & 2 (0,1,0) & 7 (1,0,0) & 1 \\
5 (1,0,1) & 3 (0,1,1) & 6 (1,1,0) & 0 \\
5 (1,0,1) & 4 (1,0,0) & 1 (0,0,1) & 3 \\
5 (1,0,1) & 5 (1,0,1) & 0 (0,0,0) & 2 \\
5 (1,0,1) & 6 (1,1,0) & 3 (0,1,1) & 5 \\
5 (1,0,1) & 7 (1,1,1) & 2 (0,1,0) & 4 \\
\hline
6 (1,1,0) & 0 (0,0,0) & 6 (1,1,0) & 0 \\
6 (1,1,0) & 1 (0,0,1) & 7 (1,0,0) & 1 \\
6 (1,1,0) & 2 (0,1,0) & 4 (1,0,0) & 6 \\
6 (1,1,0) & 3 (0,1,1) & 5 (1,1,0) & 7 \\
6 (1,1,0) & 4 (1,0,0) & 2 (0,1,0) & 4 \\
6 (1,1,0) & 5 (1,0,1) & 3 (0,1,1) & 5 \\
6 (1,1,0) & 6 (1,1,0) & 0 (0,0,0) & 2 \\
6 (1,1,0) & 7 (1,1,1) & 1 (0,0,1) & 3 \\
\hline
7 (1,1,1) & 0 (0,0,0) & 7 (1,0,0) & 1 \\
7 (1,1,1) & 1 (0,0,1) & 6 (1,1,0) & 0 \\
7 (1,1,1) & 2 (0,1,0) & 5 (1,1,0) & 7 \\
7 (1,1,1) & 3 (0,1,1) & 4 (1,0,0) & 6 \\
7 (1,1,1) & 4 (1,0,0) & 3 (0,1,1) & 5 \\
7 (1,1,1) & 5 (1,0,1) & 2 (0,1,0) & 4 \\
7 (1,1,1) & 6 (1,1,0) & 1 (0,0,1) & 3 \\
7 (1,1,1) & 7 (1,1,1) & 0 (0,0,0) & 2 \\
 \end{tabular}

\end{tabular}
\caption{$g(x,m)$ for $\mathsf{N}=3$, $k=2$ }
\label{tab:Caesar1}
\end{table}
\\

\begin{figure}[!ht]
    \centering
 \subfigure[$\mathsf{N}=3$, $k=1$]{\label{fig:CBCenc}
        \includegraphics[scale=0.2]{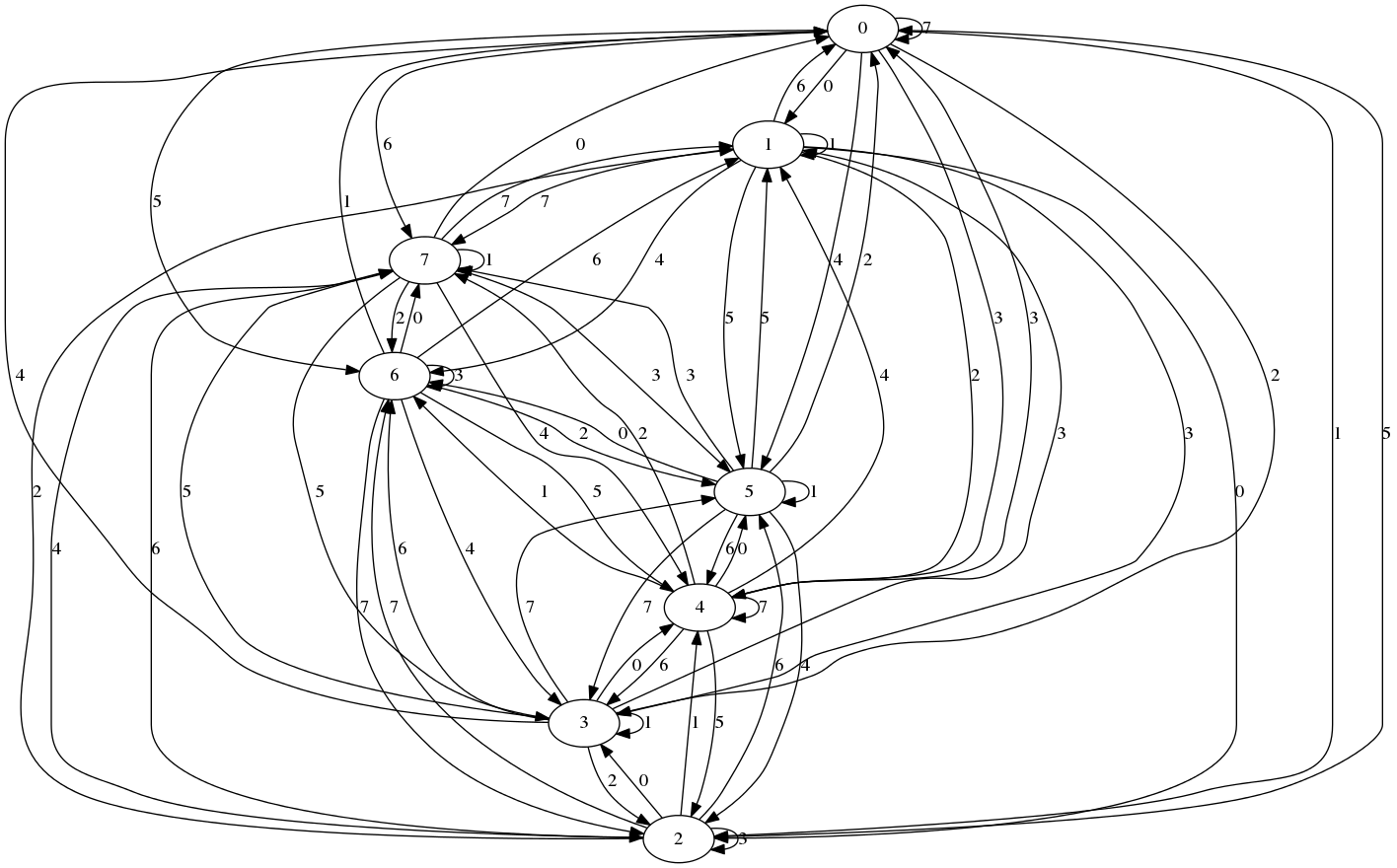}}     \subfigure[$\mathsf{N}=3$, $k=2$]{\includegraphics[scale=0.2]{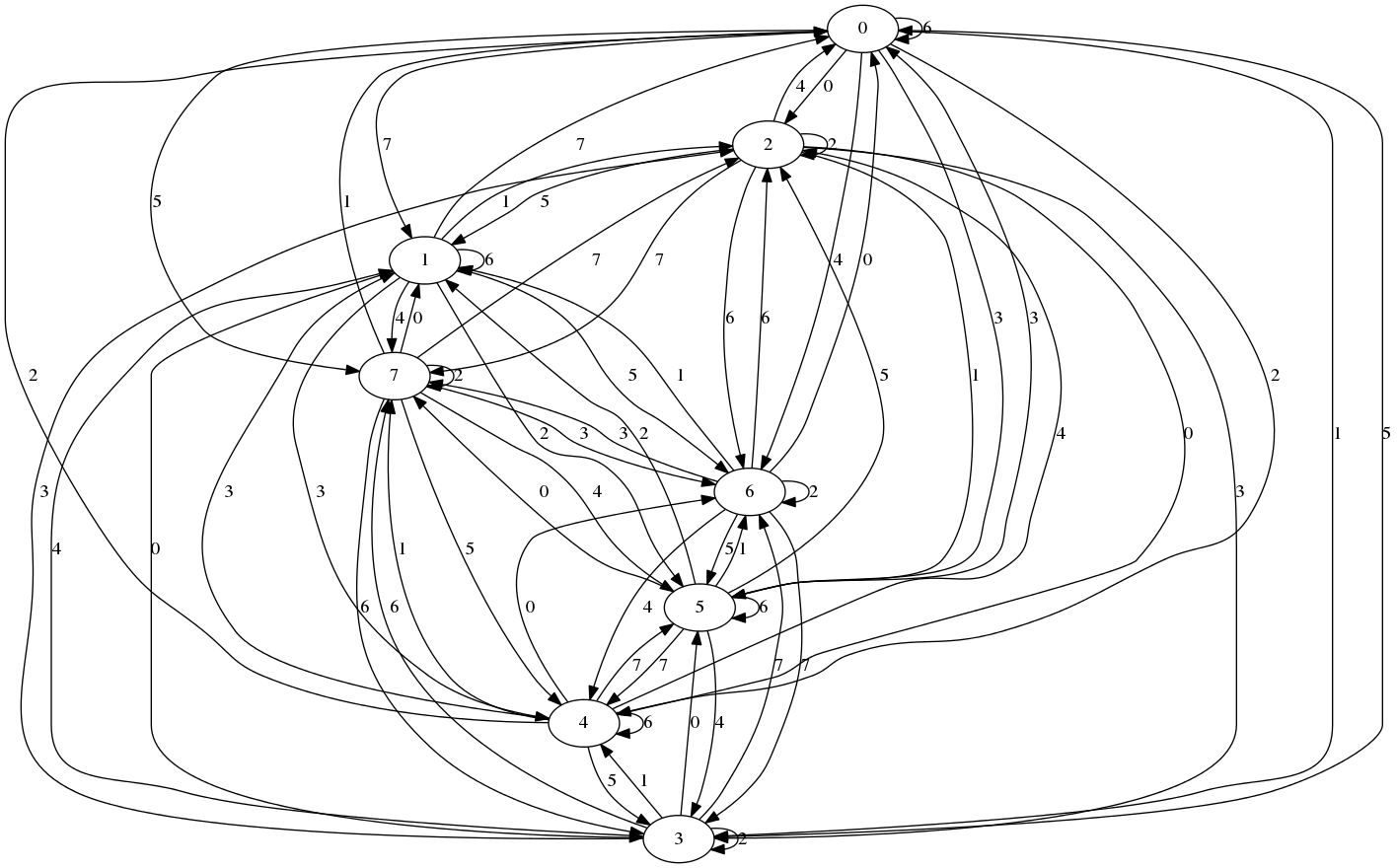}}
    \caption{$\mathcal{G}_g$ of some Caesar ciphers $\mathcal{E}_k(x):\llbracket 0, 2^\mathsf{N}-1 \rrbracket \longrightarrow \llbracket 0, 2^\mathsf{N}-1 \rrbracket , x \longmapsto x+k \mod 2^\mathsf{N}$}
     \label{fig:CBC2}
     \end{figure}

Figure~\ref{fig:CBC2}, for its part, presents the graph of iteration of the Caesar based CBC mode of operation, with the same kind of shifts, operating on blocks of size 3. 
We can verify that, at each time, the cipher block chaining behaves chaotically. In that situation, we can guarantee that any error on the IV (starting state) or on the message to encrypt (edges to browse) may potentially lead to a completely different list of visited states, that is, of a completely different ciphertext.

In this section,  it has been proven that some well chosen block ciphers can lead to a chaotical behavior for the CBC mode of operation. In the following section, we will recall some quantitative measures of chaos that have already been proven in one of our previous research work.
\section{Quantitatives measures}
\label{section:Quantitative measures}

In this section, we are interested to recall our previous results related to quantitative measures. They have been detailed in~\cite {abidi:hal-01264170}, in which both expansivity and sensibility of symmetric cyphers are regarded, in the case of CBC mode of operation. These quantitative topology metrics, taken from the mathematical theory of chaos, allow to measure in which extent a slight error on the initial condition is magnified during iterations.

In this research work we stated that, in addition to being chaotic as defined in the Devaney's formulation, the CBC mode of operation is indeed largely sensible to initial errors or modifications on either the IV or the message to encrypt. The second important tool that reinforces the chaotic behavior of the CBC mode of operation is the expansivity. This property has been evaluated too, but it is not satisfied, as it has been established thanks to a counter example. For more details and
to do this, we have began by announcing and then proving these two following propositions:
\begin{proposition}
 The CBC mode of operation is sensible to the initial condition, and its constant of sensibility is larger than the length $\mathsf{N}$ of the block size.
 \end{proposition}
 \begin{proof}
Let $X=(x;(m^0, m^1, ...)) \in \mathcal{X}$ and $\delta >0$. We are looking for $X'=(x';({m'}^0, {m'}^1, ...)$ and $n \in \mathds{N}$ such that $d(X,X') < \delta$ and $d(G_g^n(X),G_g^n(X'))>N$.

Let us define $k_0 = \lfloor -log_{10}(\delta) \rfloor +1$, in such a way that all $X'$ of the form: $$(X_1, (m^0, m^1, ..., m^{k_0}, m'^{k_0+1}, m'^{k_0+2}, ...) )$$
are such that $d(X, X')<\delta$. In other words, all messages $m'$ whose $k_0$ first blocks are equal to $(m^0, m^1, ..., m^{k_0})$ are $\delta$-close to $X$.

\begin{figure}[!h]
\centering
\includegraphics[scale=0.35]{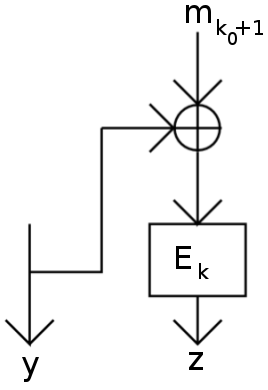}
\caption{$k_0+1$-th iterate of $G_g$}
\label{fig:iterate}
\end{figure}

Let $y=G_g^{k_0}(X)_1$ and $z=G_g^{k_0+1}(X)_1$ as defined in Figure~\ref{fig:iterate}. We consider the block message $m'$ defined by:$$m'=y\oplus \mathcal{D}_k(\overline{z})$$ 
where $\mathcal{D}_k$ is the keyed decryption function associated to $\mathcal{E}_k$, and $\overline{z}$ is the negation of $z$. We thus define $X'$ as follow:
\begin{itemize}
\item $X_1'=x$,
\item $\forall k \leqslant k_0, {m'}^k=m^k$,
\item ${m'}^{k_0+1}=m'$,
\item $\forall k \geqslant k_0+2$, ${m'}^k=\overline{m^k}$,
\end{itemize}
so $d(G_g^{k_0+1}(X), G_g^{k_0+1}(X'))$

$\begin{array}{l}
= d\left(G_g\left(y;(m_{k_0+1}, m_{k_0+2}, ...)\right),\right.\\        
~~~~~~ \left. G_g\left(y;(m', \overline{m_{k_0+1}}, \overline{m_{k_0+2}}, ...)\right) \right)\\
= d\left(\left(z;(m_{k_0+2}, m_{k_0+3}, ...)\right),\right.\\        
~~~~~~ \left. \left(E_k(y\oplus m');(\overline{m_{k_0+1}}, \overline{m_{k_0+2}}, ...)\right) \right)\\
= d\left(\left(z;(m_{k_0+2}, m_{k_0+3}, ...)\right),\right.\\        
~~~~~~ \left. \left(E_k(y\oplus (y \oplus D_k(\overline{z})));(\overline{m_{k_0+1}}, \overline{m_{k_0+2}}, ...)\right) \right)\\
= d\left(\left(z;(m_{k_0+2}, m_{k_0+3}, ...)\right),\right.\\        
~~~~~~ \left. \left(E_k((y\oplus y) \oplus D_k(\overline{z}));(\overline{m_{k_0+1}}, \overline{m_{k_0+2}}, ...)\right) \right)\\
= d\left(\left(z;(m_{k_0+2}, m_{k_0+3}, ...)\right),\right.\\        
~~~~~~ \left. \left(E_k( 0 \oplus D_k(\overline{z}));(\overline{m_{k_0+1}}, \overline{m_{k_0+2}}, ...)\right) \right)\\
= d\left(\left(z;(m_{k_0+2}, m_{k_0+3}, ...)\right),\right.\\        
~~~~~~ \left. \left(E_k(D_k(\overline{z}));(\overline{m_{k_0+1}}, \overline{m_{k_0+2}}, ...)\right) \right)\\
= d\left(\left(z;(m_{k_0+2}, m_{k_0+3}, ...)\right), \left(\overline{z};(\overline{m_{k_0+1}}, \overline{m_{k_0+2}}, ...)\right) \right)\\
= d_e(z, \overline{z})\\
~~~~~~ + d_m((m_{k_0+2}, m_{k_0+3}, ...), (\overline{m_{k_0+1}}, \overline{m_{k_0+2}}, ...) )\\
= \mathsf{N} + \dfrac{9}{\mathsf{N}} \sum_{k=k_0+2}^\infty \dfrac{m_k-\overline{m_k}}{10^k}\\
= \mathsf{N} + \dfrac{9}{\mathsf{N}} \sum_{k=k_0+2}^\infty \dfrac{\mathsf{N}}{10^k}\\
= \mathsf{N} + 9 \sum_{k=k_0+2}^\infty \left(\dfrac{1}{10^k}\right) = \mathsf{N} + \dfrac{1}{10^{k_0+1}} > \mathsf{N},
\end{array} $

\noindent which concludes the proof of the sensibility of $G_g$.
\end{proof}
\begin{proposition}
The CBC mode of operation is not expansive.
\end{proposition}
\begin{proof} 
Consider for instance two initial vectors $x=(1,0,\hdots, 0)$ and $x'=(0,1,0,\hdots, 0)$, associated to the messages $m=((0,1,0, \hdots, 0), (0, \hdots, 0), (0, \hdots, 0), \hdots )$ and $m'=((1,0, \hdots, 0), (0, \hdots, 0), (0, \hdots, 0), \hdots )$: all blocks of messages are null in both $m$ and $m'$, except the first block. Let $X=(x,m)$ and $X'=(x',m')$.

Obviously, $x \neq x'$, while $x \oplus m_0 = x' \oplus m_0'$. This latter implies that $X_1^0 = {X'}_1^0$, and by a recursive process, we can conclude that $\forall i \in \mathds{N}, X_1^i = {X'}_1^i$. So the distance between points $X=(x,m)$ and $X'=(x',m')$ is strictly positive, while for all  $n>0$, $d\left(G_g^n(X), G_g^n(X')\right)=0$, which concludes the proof of the non expansive character of the CBC mode of operation by the mean of the exhibition of a counter example.
\end{proof}
To sum up, proving these two propositions claimed previously allowed us to conclude that CBC mode of operation is sensible to the initial conditions. But, on the other side, it is not expansive. 

Let us now investigate new original aspects of chaos of the CBC mode of operation.

\section{Topological mixing and Topological entropy}
\label{top mixing}
\subsection{Topological mixing}
\label{topo}
As mentioned in Definition~\ref{def:topMixing}, a discrete dynamical system is said \emph{topologically mixing} if and only if, for any couple of disjoint open set $\mathcal{U},\mathcal{V} \neq \varnothing$, there exists an integer $n_0\in \mathds{N}$ such that, for all $n > n_0$, $f^{\circ n}(\mathcal{U}) \cap \mathcal{V} \neq
\varnothing$.
In ~\cite{abidi:hal-01312476}, we have deepened the topological study for the CBC mode of operation. Indeed, we have regarded if this tool possesses the property of topological mixing. In order to proof this property, we have began by stating the following proposition: 
\begin{proposition}
\label{prop:topological mixing}
$(\mathcal{X}, G_g)$ is topologically mixing.
\end{proposition}
which is an immediate consequence of the lemma below.
\begin{lemma}
For any open ball  $\mathcal{B}=\mathcal{B}((x,m),\varepsilon)$ of $\mathcal{X}$, an index $n$ can be found such that $G_{g}^{\circ n}(\mathcal{B}) = \mathcal{X}$.
\end{lemma}
Proving this lemma led us to conclude the proposition claimed previously. Hence, this dynamical system owns the property of topological mixing.

In addition to this property, other quantitative evaluations have been performed, and the level of topological entropy has been evaluated too.

\subsection{Topological entropy}
 \label{sec:entro}
Another important tool to measure the chaotic behavior of a dynamical system is the topological entropy, which is defined only for compact topological spaces. Therefore, before studying the entropy of CBC mode of operation, we must then check that $(\mathcal{X},\ d)$ is compact.
\subsubsection{Compactness study}

In this section, we will prove that $(\mathcal{X}, d)$ is a compact topological space, in order to study its topological entropy later. Firstly, as $(\mathcal{X}, d)$ is a metric space, it is separated.
 It is however possible to give a direct proof of this result:

\begin{proposition}
$(\mathcal{X}, d)$ is a separated space.
\end{proposition}

\begin{proof}
Let $(x,w) \neq (\textrm{\^{x}},\textrm{\^{w}})$ two points of $\mathcal{X}$.

\begin{enumerate}
\item If $x \neq \textrm{\^{x}}$, then the intersection between the two balls  $\mathcal{B}\left((x,w),\frac{1}{2}\right)$ and $\mathcal{B}\left((\textrm{\^{x}},\textrm{\^{w}}), \frac{1}{2}\right)$  is empty.
\item Else, it exists $k\in\mathds{N}$ such that $w_k \neq \textrm{\^{w}}_k$, then the balls $\mathcal{B}\left((x,w),10^{-(k+1)}\right)$ and $\mathcal{B}\left((\textrm{\^{x}},\textrm{\^{w}}), 10^{-(k+1)}\right)$ can be chosen.
\end{enumerate}
\end{proof}
Let us now prove the compactness of the metric space $(\mathcal{X}, d)$ by using the sequential characterization of compactness.

\begin{proposition}
$(\mathcal{X}, d)$ is a compact space.
\end{proposition}
\subsubsection{Topological entropy}
Let $(X, d)$ be a compact metric space and $f: X \rightarrow X$ be a continuous map. For each natural number $n$, a new metric $d_n$ is defined on $X$ by

$$d_n(x,y)=\max\{d(f^{\circ i}(x),f^{\circ i}(y)): 0\leq i<n\}.$$

Given any $\varepsilon > 0$ and $n \geqslant 1$, two points of $X$ are $\varepsilon$-close with respect to this new metric if their first $n$ iterates are $\varepsilon$-close (according to $d$).

This metric allows one to distinguish in a neighborhood of an orbit the points that move away from each other during the iteration from the points that travel together. A subset $E$ of $X$ is said to be $(n, \varepsilon)$-separated if each pair of distinct points of $E$ is at least $\varepsilon$ apart in the metric $d_n$.

\begin{definition}
Let $H(n, \varepsilon)$ be the maximum cardinality of a $(n, \varepsilon)$-separated set, the \emph{topological entropy} of the map $f$ is defined by (see \textit{e.g.},~\cite{Adler65} or~\cite{Bowen})
$$h(f)=\lim_{\epsilon\to 0} \left(\limsup_{n\to \infty} \frac{1}{n}\log H(n,\varepsilon)\right). $$
\end{definition}

We have the result,

\begin{theorem}
Entropy of $(\mathcal{X},G_g)$ is infinite.
\end{theorem}
\begin{proof}
Let $\textrm{x}, \textrm{\v{x}}\in \mathbb{B}^\mathsf{N}$ such that $\exists i_0 \in \llbracket 1, N \rrbracket, \textrm{x}_{i_0} \neq \textrm{\v{x}}_{i_0}$. Then, $\forall \textrm{w}, \textrm{\v{w}} \in \mathcal{S}_\mathsf{N}$,
$$d((\textrm{x},\textrm{w});(\textrm{\v{x}},\textrm{\v{w}})) \geqslant 1$$
But the cardinal $c$ of $\mathcal{S}_\mathsf{N}$ is infinite, then $\forall n \in \mathbb{N}, c >e^{n^2}$.

So for all $n \in \mathbb{N}$, the maximal number $H(n,1)$ of $(n,1)-$separated points is greater than or equal to $e^{n^2}$, and then
$$
\begin{array}{ll}
h_{top}(G_g,1) & = \overline{lim} \frac{1}{n} log \left( H(n,1)\right) > \overline{lim} \frac{1}{n} log \left( e^{n^2} \right) \\
&= \overline{lim} ~(n) = + \infty.
\end{array}$$
\noindent But $h_{top}(G_g,\varepsilon)$ is an increasing function when  $\varepsilon$ is decreasing, then
$$
\begin{array}{ll}
h_{top} \left( G_g \right) &= \lim_{\varepsilon \rightarrow 0} h_{top}(G_g,\varepsilon) \\
&> h_{top}(G_g,1) = + \infty,
\end{array}$$
\noindent which concludes the evaluation of the topological entropy of $G_g$.
\end{proof}

In conclusion, all of these properties lead to a complete unpredictable behavior for the CBC mode of operation.
\section{Conclusion}
In this paper, our goal was to summarize our numerous results that have been detailed in our previous series of articles. Hence, the interest of our work is not to provide a collection of secure and complex CBC, but to initiate a complementary approach for studying such modes of operation. Our intention is to show how to model such modes, and that it is possible to study the complexity of their dynamics. Up-to-date block ciphers and modes of operation, together with topological analyses using most recent developments in this field, need to be investigated, while the interest of each topological property of complexity must be related to desired objectives for each mode of operation.
 




%
\bibliographystyle{unsrt}
\bibliography{biblio}
\end{document}